\documentclass[10pt,journal]{amsart}

\usepackage{amsmath}
  \usepackage{paralist}
  \usepackage{graphics} 
  \usepackage{epsfig} 
\usepackage{graphicx}  \usepackage{epstopdf}
 \usepackage[colorlinks=true]{hyperref}
\hypersetup{urlcolor=blue, citecolor=red}

\newtheorem{theorem}{Theorem}[section]
\newtheorem{corollary}{Corollary}

\newtheorem{lemma}[theorem]{Lemma}
\newtheorem{proposition}{Proposition}

\theoremstyle{definition}
\newtheorem{definition}[theorem]{Definition}

\newtheorem{example}{Example}

\usepackage{amsmath,amssymb,color}

\newcommand{\Cc}{\mathcal{C}}
\newcommand{\Ic}{\mathcal{I}}

\newcommand{\Oc}{\mathcal{O}}

\newcommand{\FF}{\mathbb{F}}
\newcommand{\QQ}{\mathbb{Q}}
\newcommand{\RR}{\mathbb{R}}
\newcommand{\ZZ}{\mathbb{Z}}

\newcommand{\pf}{\mathfrak{p}}

\title[Skew Polynomial Codes from Quotients of Division Algebras]{On Skew Polynomial Codes and Lattices from Quotients of Cyclic Division Algebras}

\author[J\'er\^ome Ducoat and Fr\'ed\'erique Oggier]{J\'er\^ome Ducoat and Fr\'ed\'erique Oggier}

\address{Division of Mathematical Sciences, School of Physical and Mathematical Sciences, \newline\indent Nanyang Technological University, Singapore.} 



 \email{jerome.ducoat@gmail.com}
 \email{frederique@ntu.edu.sg}

\thanks{
The research of J. Ducoat and F. Oggier was supported by the Singapore National Research Foundation under Research Grant NRF-RF2009-07.
}

\begin{document}
\maketitle
\begin{abstract}
We propose a variation of Construction A of lattices from linear codes defined using the quotient $\Lambda/\pf\Lambda$ of some order $\Lambda$ inside a cyclic division $F$-algebra, for $\pf$ a prime ideal of a number field $F$. To obtain codes over this quotient, we first give an isomorphism between $\Lambda/\pf\Lambda$ and a ring of skew polynomials.
We then discuss definitions and basic properties of skew polynomial codes, which are needed for Construction A, but also explore further properties of the dual of such codes. We conclude by providing an application to space-time coding, which is the original motivation to consider cyclic division $F$-algebras as a starting point for this variation of Construction A.
\end{abstract}

%
%
%
\section{Introduction}
Connections between Euclidean lattices and linear codes over finite alphabets have been classically studied, through a series of constructions referred to as Constructions A,B,C,D or E \cite{splag}. Given the finite field $\FF_2$, the original (binary) Construction A \cite{Forney} considers the map $\rho :\ZZ^N \rightarrow \FF_2^N$ of reduction modulo 2 componentwise, for $N$ some positive integer. If $C$ is a binary linear code of length $N$, then $\rho^{-1}(C)$ (or its normalized version $\frac{1}{\sqrt 2}\rho^{-1}(C)$) is a lattice. A series of duality follows, such as a correspondence between the dual code $C^\perp$ of $C$ and the dual lattice $L^*$ of $L$, self-dual codes and unimodular lattices, or between the weight enumerator of the
code $C$ and the theta series of the lattice $L$. Construction A has been generalized in many ways, including, for citing a few examples, to: (1) cyclotomic fields $\QQ(\zeta_p)$ \cite{Ebeling}, where $\zeta_p$ is a primitive $p$th root of unity, and $p$ is prime, in which case $\rho$ becomes the componentwise reduction modulo the prime ideal $(1 - \zeta_p)$ of a vector with coefficients in the ring of integers $\ZZ[\zeta_p]$, (2) to quadratic and imaginary fields and totally definite quaternion algebras over $\QQ$ as a way to construct modular lattices in \cite{Bachoc}, (3) to the construction of unimodular lattices using quaternary quadratic residue codes \cite{BSC}.

Construction A is also of interest from a communication point of view, since it provides a concrete way of doing coset coding \cite{Forney}, an encoding technique which is very useful to encode lattices, since it provides a natural way of mapping elements from finite alphabets (codewords) to real or complex symbols (lattice points).

This paper addresses a variation of Construction A in the context of division algebras, and cyclic division algebras over a cyclic number field extension $K/F$ in particular. Denote by $\Oc_F$ and $\Oc_K$ their respective rings of integers, and by $\sigma$ the generator of the Galois group of $K/F$. Instead of using as a starting point
the quotient of a number field by an ideal (as in \cite{Ebeling} for example), we consider a specific $\Oc_F$-order $\Lambda$ in a cyclic division $F$-algebra, and the quotient $\Lambda/\pf\Lambda$ for $\pf$ a prime ideal of $\Oc_F$. 

As shown in Section \ref{sec:quot}, the quotient $\Lambda/\pf\Lambda$ turns out to be isomorphic to the ring of skew polynomials $(\Oc_K/\pf\Oc_K)[x; \sigma]/(x^n - u)$, where $u$ is determined by the choice of the cyclic division $F$-algebra. To mimic the steps involved in Construction A, we next consider the problem of designing codes over $(\Oc_K/\pf\Oc_K)[x; \sigma]/(x^n - u)$ in Section \ref{sec:code}, where a variation of Construction A is given. Let $\FF_q$ denote the finite field with $q$ elements, $q$ a prime power. Codes over the skew polynomial ring $\FF_q[x; \sigma]/(f(x))$, for $(f(x))$ a two sided ideal of $\FF_q[x; \sigma]$, were introduced in \cite{Boucher} and furthermore studied in the context of cyclic codes in \cite{BGU}. A generalization to Galois rings was obtained in \cite{BUS}. We propose some basic definitions and properties of codes over the skew polynomial ring $(\Oc_K/\pf\Oc_K)[x;\sigma]/(x^n - u)$, by taking inspiration from \cite{BUS} and generalizing some of the known results for the ring $\FF_q[x;\sigma]/(x^n - u)$. We note that our results could be easily generalized to the case $(\Oc_K/\pf\Oc_K)[x;\sigma]/(f(x))$, which is not done in this paper because of our focus on cyclic division algebras.

Once equipped with suitable definitions for skew polynomial codes over the ring $(\Oc_K/\pf\Oc_K)[x;\sigma]/(x^n - u)$, it is natural to wonder about the dual of such codes.  
This is studied in Section \ref{sec:dualcode}.  

We conclude by going back to the original motivation to consider this variation of Construction A over a cyclic division $F$-algebra, namely its application to space-time coding, or more precisely, to derive a way to perform coset encoding, as detailed in Section \ref{sec:matrix}.

%
%
%
\section{Quotients of Cyclic Division Algebras}
\label{sec:quot}

Let $K/F$ be a number field extension of degree $n$ with cyclic Galois group $\langle\sigma \rangle$, and respective rings of integers $\Oc_K$ and $\Oc_F$.
Consider the cyclic $F$-algebra $A$ defined by
\[
K \oplus Ke \oplus \cdots Ke^{n-1}
\]
where $e^n=u \in F$, and $ek=\sigma(k)e$ for $k\in K$. 
We assume that $u^i$, $i=0,\ldots,n-1$, are not norms in $K/F$ so that the algebra is division.

For $S$ a Noetherian integral domain with quotient field $F$, and $A$ a finite dimensional $F$-algebra, an $S$-order in $A$ is a subring $\Lambda$ of $A$, having the same identity as $A$, and such that $\Lambda$ is a finitely generated module over $S$ and generates $A$ as a linear space over $F$. An order $\Lambda$ is called maximal if it is not properly contained in any other $S$-order.

Then
\[
\Lambda=\Oc_K \oplus \Oc_Ke\oplus \cdots \oplus \Oc_Ke^{n-1}
\]
is an $\Oc_F$-order of $A$, which is typically not maximal.

Let $\pf$ be a prime ideal of $\Oc_F$ so that $\pf\Lambda$ is a two-sided ideal of $\Lambda$, 
and $\Oc_F/\pf\Oc_F$ is the finite field $\FF_{p^{f}}$, where $p$ is the prime number lying below $\pf$ and $f$ is the inertial degree of $\pf$ above $p$.
Since $\Lambda$ is an $\Oc_F$-module, we have the following ring homomorphism :
\[ 
\Oc_F \rightarrow \Lambda \rightarrow \Lambda / \mathfrak p\Lambda
\]  
and the image of $\mathfrak p =\mathfrak p\Lambda \cap \Oc_F$ obviously lies in $\mathfrak p \Lambda$. Hence, it yields a ring homomorphism $\FF_{p^f}=\Oc_F/\mathfrak p \rightarrow \Lambda/\mathfrak p\Lambda$, which means that $\Lambda/\pf\Lambda$ is an $\FF_{p^f}$-algebra.

From \cite{OS}, we have the following $\FF_{p^f}$-algebra isomorphism :
\[
\Lambda/\pf\Lambda \simeq (\Oc_K/\pf\Oc_K)\oplus(\Oc_K/\pf\Oc_K)e\oplus \cdots \oplus (\Oc_K/\pf\Oc_K)e^{n-1},
\]
where $e(k+\pf\Oc_K)=(\sigma(k)+\pf\Oc_K)e$ for all $k \in \Oc_K$ and $e^n=u+\pf$.

Indeed, since $\mathfrak p$ is a prime ideal of $\Oc_F$ and $\sigma_{|F}=\textrm{Id}_F$, we have $\sigma(\mathfrak p\Oc_K)\subset \mathfrak p \Oc_K$, which means that $\sigma$ can be factorized as a ring homomorphism by the natural projection $\pi:\Oc_K \rightarrow \Oc_K/\mathfrak p\Oc_K$ : more explicitely, there exists a ring homomorphism $\overline{\sigma}:\Oc_K/\mathfrak p\Oc_K \rightarrow \Oc_K/\mathfrak p\Oc_K$ such that $\sigma=\overline{\sigma}\circ \pi$.
By a slight abuse of notation, we keep the notation $\sigma$ for the map $\overline{\sigma}$. 

Since $\Oc_K$ is a Dedekind domain, $\pf\Oc_K=\pf_1^{e_1}\pf_2^{e_2}\cdots\pf_g^{e_g}$ with $\pf_i$ a prime ideal of $\Oc_K$ and $e_i\geq 0$ for $i=1,\ldots,g$.

In the particular case where $\pf$ is inert in $K/F$, $\pf\Oc_K$ is a prime ideal of $\Oc_K$. Then the finite ring $\Oc_K/\pf\Oc_K$ is an integral domain, so it is the finite field $\FF_{p^{nf}}$. The induced ring homomorphism $\sigma:\FF_{p^{nf}}=\Oc_K/\pf\Oc_K \rightarrow \Oc_K/\pf\Oc_K=\FF_{p^{nf}}$ is thus a generator of the cyclic Galois group of $\FF_{p^{nf}}/\FF_{p^f}$. 

Given a ring $S$ with a group $G=\langle \sigma \rangle$ acting on it, the skew polynomial ring $S[x;\sigma]$ is the set of all polynomials $s_0+s_1x+\ldots+ s_mx^m$, $s_i \in S$, $m\geq 0$ with multiplication twisted by the relation $xs=\sigma(s)x$ for all $s \in S$.

We will consider the skew polynomial ring $(\Oc_K/\pf\Oc_K)[x;\sigma]$.
Note that since $u\in F$, $x^n-u$ belongs to the center of $(\Oc_K/\pf\Oc_K)[x;\sigma]$ and the ideal $(x^n-u)$ is two-sided.

\begin{lemma}
\label{l1}
There is an $\FF_{p^f}$-algebra isomorphism between $\Lambda/\pf\Lambda$ and the quotient of $(\Oc_K/\pf\Oc_K)[x;\sigma]$ by the two-sided ideal generated by $x^n-u$.
\end{lemma}

\begin{proof}
We define the map 
\[
\begin{aligned}
\varphi : (\Oc_K/\pf\Oc_K)[x;\sigma] &\rightarrow  \Lambda/\pf\Lambda \\ f(x) &\mapsto f(e).\end{aligned}
\]
Using the isomorphism given above and in \cite{OS}, it is easily seen that $\varphi$ is a surjective $\FF_{p^f}$-algebra homomorphism.
Moreover, the kernel of $\varphi$ is the two-sided ideal of $(\Oc_K/\pf\Oc_K)[x;\sigma]$ generated by $x^n-u$. Indeed, it is easily seen that $x^n-u$ lies in $\rm{ker}(\varphi)$. Conversely, let $f(x) \in \rm{\ker}(\varphi)$. We write 
\[
f(x)=\underset{i=0}{\overset{m}{\sum}} s_ix^i,~s_i\in \Oc_K/\pf\Oc_K,~i=0,...,m.
\] 
Then $f(e)=0$ in $\Lambda/\pf\Lambda$. 
The ring $(\Oc_K/\pf\Oc_K)[x;\sigma]$ is not always right or left Euclidean, as it would be if $\Oc_K/\pf\Oc_K$ were a finite field.
However, since $x^n-u$ is a monic polynomial (or more precisely since its leading coefficient is a unit in $\Oc_K/\pf\Oc_K$), we can still perform the long division algorithm: indeed, if $m \geq n$, we may perform a right division of $f(x)$ by $x^n-u$ as follows. Note that the polynomial
\[
f(x)-s_m x^{m-n}(x^n-u)=\underset{i=0}{\overset{m-1}{\sum}} s_ix^i+s_mux^{m-n}
\]
has degree smaller than that of $f$. This procedure of subtracting left multiple of $x^n-u$ can be repeated until we obtain a polynomial of degree smaller than $n$, thus, there exist some polynomials $g(x)$ and $h(x)$ such that 
\[
f(x)=g(x)(x^n-u)+h(x)
\] 
where $h(x)$ has degree $\leq n-1$. Hence, $f(e)=0$ is equivalent to $h(e)=0$. Yet, $0=h(e)=r_0+r_1e+\cdots +r_{n-1}e^{n-1}$ in $\Lambda/\pf\Lambda\simeq (\Oc_K/\pf\Oc_K)\oplus(\Oc_K/\pf\Oc_K)e\oplus \cdots \oplus (\Oc_K/\pf\Oc_K)e^{n-1}.$ Therefore, $r_0=r_1=\cdots = r_{n-1}=0$ and $h(x)=0$. We conclude that $f(x)$ is a (left) multiple of $x^n-u$. Consequently, ${\rm{ker}}(\varphi)=(x^n-u)$ and we get the desired isomorphism. 
\end{proof}

This lemma generalizes the isomorphism of \cite[Lemma 1]{castle} for $\pf$ inert.

Denote by $\psi$ the inverse isomorphism of the one given in Lemma \ref{l1}: 
\[
\psi: \Lambda/\pf\Lambda \cong (\Oc_K/\pf\Oc_K)[x;\sigma]/(x^n-u).
\]

Let $\Ic$ be a left ideal of $\Lambda$. Assume that $\Ic\cap \Oc_F\supset \pf$. Then $\Ic/\pf\Lambda$ is an ideal of $\Lambda/\pf\Lambda$. In the sequel, we will study the left ideal $\psi(\Ic/\pf\Lambda)$ of $(\Oc_K/\pf\Oc_K)[x;\sigma]/(x^n-u)$.

%
%
%

\section{Skew-Polynomial Codes and a Variation of Construction A}
\label{sec:code}

Classical cyclic codes over the finite field $\FF_q$ are ideals of the polynomial ring $\FF_q[x]/(x^n-1)$. Since this is a principal ideal domain, every ideal has a generator polynomial, which in turn defines a cyclic code, and codewords are then multiples of this generator polynomial.
In \cite{Boucher}, the definition of codes over polynomial rings was extended to that of codes over the skew-polynomial rings $\FF_q[x;\sigma]$. We are interested here in codes over the skew-polynomial rings $(\Oc_K/\pf\Oc_K)[x;\sigma]/(x^n-u)$.\\

\subsection{Division in Skew-Polynomial Rings.}
A first major difference between $\FF_q[x;\sigma]$ and $(\Oc_K/\pf\Oc_K)[x;\sigma]$ is, as noted in the proof of Lemma \ref{l1}, that the skew polynomial ring $(\Oc_K/\pf\Oc_K)[x;\sigma]$ is not in general right or left Euclidean anymore. However, this can be sorted out, using the same technique as that explained in \cite{BUS}, which was also used for the proof of Lemma \ref{l1}.

Let $f(x)=\underset{i=0}{\overset{m}{\sum}} s_ix^i$ and $g(x)=\underset{i=0}{\overset{l}{\sum}} t_ix^i$ be two polynomials of respective degree $m$ and $l$, with coefficients in $(\Oc_K/\pf\Oc_K)$. 
Suppose that $t_l$ is invertible and $m>l$, then a right (respectively left) division of $f(x)$ by $g(x)$ is obtained as follows: consider the polynomial
\[
f(x)-\frac{s_m}{\sigma^{m-l}(t_l)}x^{m-l}g(x)
\]
respectively
\[
f(x)-g(x)\sigma^{-l}\left(\frac{s_m}{t_l}\right)x^{m-l}.
\]
Since $\sigma$ is an automorphism, $1=\sigma(t_lt_l^{-1})=\sigma(t_l)\sigma(t_l^{-1})$, which implies that $\sigma(t_l)$ is invertible. Hence, $\sigma^j(t_l)$ also is invertible for all $1\leq j\leq m$. Moreover, we have 
\[ \begin{aligned} f(x)-\frac{s_m}{\sigma^{m-l}(t_l)}x^{m-l}g(x) &= f(x)-s_m x\sigma^{m-l-1}(t_l^{-1}) x^{m-l} g(x) = \cdots \\ &= f(x)-s_m x^{m-l} \sigma^{m-l-(m-l)}(t_l^{-1}) g(x)\\ &=f(x)-s_m x^{m-l} t_l^{-1} (t_l x^{l} + \underset{i=0}{\overset{l-1}{\sum}} t_ix^i)\\ &=f(x)-s_m x^{m} + \underset{i=0}{\overset{l-1}{\sum}}  s_m t_l^{-1}t_ix^i.\end{aligned}
\]
Similarly, we have
\[ \begin{aligned} f(x)-g(x) \sigma^{-l}\left(\frac{s_m}{t_l}\right) x^{m-l} &= f(x)- \underset{i=0}{\overset{l}{\sum}} t_ix^i \sigma^{-l}\left(\frac{s_m}{t_l}\right) x^{m-l} \\ &= f(x)- \underset{i=1}{\overset{l}{\sum}} t_ix^{i-1} \sigma^{1-l}\left(\frac{s_m}{t_l}\right) x^{m-l+1} - t_0\sigma^{-l}\left(\frac{s_m}{t_l}\right)x^{m-l}\\ & = \cdots \\ & =f(x)- \underset{i=0}{\overset{l}{\sum}} t_i \sigma^{i-l}\left(\frac{s_m}{t_l}\right) x^{m-l+i} 
\\ & =
f(x)-s_m x^m - \underset{i=0}{\overset{l-1}{\sum}} t_i \sigma^{i-l}\left(\frac{s_m}{t_l}\right) x^{m-l+i}.\end{aligned}
\]
Hence, both polynomials $f(x)-\frac{s_m}{\sigma^{m-l}(t_l)}x^{m-l}g(x)$ and $f(x)-g(x) \sigma^{-l}\left(\frac{s_m}{t_l}\right) x^{m-l}$ have degree $< m$. Therefore this procedure can be iterated, to find polynomials $q(x),r(x)$ such that
\[
f(x)=q(x)g(x)+r(x),~\mbox{resp. } f(x)=g(x)q(x)+r(x),~\deg (r(x)) < \deg (g(x)).
\]
If $r(x)=0$, we say that $g(x)$ is a right, resp. left, divisor of $f(x)$.

We next discuss the unicity of the remainder of this division. Suppose 
\[
f(x)=q_1(x)g(x)+r_1(x)=q_2(x)g(x)+r_2(x)
\]
then $(q_1(x)-q_2(x))g(x)=r_2(x)-r_1(x)$. Suppose that $q_1(x)-q_2(x)$ is not zero, then the degree of 
$(q_1(x)-q_2(x))g(x)$ is greater than $\deg(g(x))$, while the degree of $r_2(x)-r_1(x)$ is less than $\deg(g(x))$, therefore $q_1(x)=q_2(x)$ and $r_1(x)=r_2(x)$.
The proof for left division is done similarly.\\

\subsection{Principal Ideals and Codes.}

Thanks to the unicity of the remainder, the skew polynomials of degree less than $n$ are canonical representatives of the elements of $(\Oc_K/\pf\Oc_K)[x;\sigma]/(x^n-u)$.

\begin{lemma}
Any right divisor $g(x)$ of $x^n-u$ generates a left principal ideal $(g(x))/(x^n-u)$ of 
$(\Oc_K/\pf\Oc_K)[x;\sigma]/(x^n-u)$. The set of left multiples of $g(x)$ by skew polynomials of degree $k=n-\deg(g(x))$ are canonical representatives of the elements of $(g(x))/(x^n-u)$. 
\end{lemma}
\begin{proof}
For any right divisor $g(x)$ of $x^n-u$, the ideal $(x^n-u)$ is contained in $(g(x))$. By the correspondence theorem for rings, $(g(x))/(x^n-u)$ is a left ideal of $(\Oc_K/\pf\Oc_K)[x;\sigma]/(x^n-u)$. Since skew polynomials of degree less than $n$ are canonical representatives, the elements of the ideal $(g(x))/(x^n-u)$ are left multiples of $g(x)$ by skew polynomials of degree less than $k=n-\deg g(x)$.
\end{proof}

\begin{corollary}
For $g(x)$ a right divisor of $x^n-u$, the ideal $(g(x))/(x^n-u)$ is an $\Oc_K/\pf\Oc_K$-module, isomorphic to a submodule of $(\Oc_K/\pf\Oc_K)^n$. It forms a code of length $n$ and of dimension $k=n-\deg g(x)$ over 
$\Oc_K/\pf\Oc_K$, 
consisting of codewords $a=(a_0,a_1,\ldots,a_{n-1})$ that are coefficient tuples of the left multiples $a(x)=a_0+a_1x+\ldots+a_{n-1}x^{n-1}$ of $g(x)$.
This code is called a {\em $\sigma$-constacyclic code}.
\end{corollary}

\begin{proof}
Let $g(x)$ be a right divisor of $x^n-u$ in $(\Oc_K/\pf\Oc_K)[x;\sigma]$. The left ideal $(g(x))$ in $(\Oc_K/\pf\Oc_K)[x;\sigma]$ is then an $\Oc_K/\pf\Oc_K$-module and by taking the quotient by the monic skew-polynomial $x^n-u$, $(g(x))/(x^n-u)$ is a submodule of the free $\Oc_K/\mathfrak p\Oc_K$-module $(\Oc_K/\pf\Oc_K)[x;\sigma]/(x^n-u)$ of rank $n$. By Lemma 1, the images of the skew polynomials $g(x), xg(x),...,x^{k-1}g(x)$ in the quotient ring $(\Oc_K/\pf\Oc_K)[x;\sigma]/(x^n-u)$ form a basis of $g(x)/(x^n-u)$ as an $\Oc_K/\mathfrak p\Oc_K$-module. We then use the $\Oc_K/\mathfrak p\Oc_K$-module isomorphism \[\begin{aligned} (\Oc_K/\pf\Oc_K)[x;\sigma]/(x^n-u) \hspace*{.5cm}& \rightarrow (\Oc_K/\pf\Oc_K)^n \\ \underset{i=0}{\overset{n-1}{\sum}} a_ix^i +(x^n-u)(\Oc_K/\pf\Oc_K)[x;\sigma] &\mapsto (a_0,...,a_{n-1})\end{aligned}\] to conclude.
\end{proof}

Note that for a $\sigma$-constacyclic code $\Cc$, corresponding to a left principal ideal $\Ic$ of $(\Oc_K/\pf\Oc_K)[x;\sigma]/(x^n-u)$, if $(a_0,\ldots,a_{n-1}) \in \Cc$, then $a(x)=a_0+a_1x+\ldots+a_{n-1}x^{n-1} \in \Ic$, and $xa(x) \in \Ic$. Since
\[
xa(x)=x(a_0+a_1x+\ldots+a_{n-1}x^{n-1})=\sigma(a_0)x+\sigma(a_1)x^2+\ldots+\sigma(a_{n-1})u
\]
we have that
\[
(u\sigma(a_{n-1}),\sigma(a_0),\ldots,\sigma(a_{n-1})) \in \Cc,
\]
which explains the choice of the terminology ``constacyclic".

Let us comment on the proposed definition of $\sigma$-constacyclic code.
\begin{enumerate}
\item
We could also use the terminology {\em central $\sigma$-code} proposed in \cite{Boucher} since $x^n-u$ lies in the center of $(\Oc_K/\pf\Oc_K)[x;\sigma]$.
\item
A cyclic code would correspond to the case $u=1$.
\item
This definition could be made more general by considering another suitable polynomial 
than $x^n-u$, which would yield a complete generalization of the definition of {\em $\sigma$-code} proposed in \cite{Boucher}. This is not done here, simply because the polynomial is given by the cyclic algebra structure, but this could be an interesting research direction of its own.
\item Finally, we may (and do) restrict our study of $\sigma$-constacyclic codes to right divisors $g(x)$ of $x^n-u$ that are monic : indeed, if $g(x)$ has leading coefficient $a$, it is necessarily invertible (since $g(x)$ is a divisor of $x^n-u$) and the constacyclic codes derived from $g(x)$ and $a^{-1}g(x)$ are equal.\\

\end{enumerate}

The next result describes a parity check polynomial.
\begin{proposition}\label{prop:parity}
Set $h(x)=\underset{i=0}{\overset{k-1}{\sum}}h_ix^i+x^k$. If $h(x)g(x)=x^n-u$, then  
\[
g(x)h(x)=x^n-u.
\]
Furthermore, let $\Cc$ be the code generated by $g(x)$. Let $a\in (\Oc_K/\mathfrak p \Oc_K)^n$ and let $a(x)$ be its corresponding polynomial.
Then $a$ is a codeword of $\Cc$ if and only if $a(x)h(x)=0$ in $(\Oc_K/\pf\Oc_K)[x;\sigma]/(x^n-u)$. 
\end{proposition}
\begin{proof}
Since $x^n-u$ is central,
\[
h(x)g(x)h(x)=(x^n-u)h(x)=h(x)(x^n-u)=h(x)h(x)g(x)
\]
and
\[
h(x)[g(x)h(x)-h(x)g(x)]=0.
\]
The first claim follows from the fact that $h(x)$ is monic and thus not a zero divisor.

If $a\in \Cc$ then $a(x)=b(x)g(x)$ and using the first claim
\[
a(x)h(x)=b(x)g(x)h(x)= 0 \mod x^n-u.
\]
Conversely, if $a(x)h(x)=0$, then for some $b(x)$, $a(x)h(x)=b(x)(x^n-u)=b(x)g(x)h(x)$. Since $h(x)$ is monic and thus not a zero divisor, $a(x)=b(x)g(x)$ as needed.
\end{proof}


\subsection{A Construction A}
Using the isomorphism $\psi$ defined in Section \ref{sec:quot}, for every left ideal $\Ic$ of $\Lambda$, we consider the $\sigma$-code $\Cc=\psi(\Ic/\pf\Lambda)$ over $\FF_q$.  

A lattice of dimension $N$ is here a discrete additive subgroup of $\RR^N$ of rank $N$ as a $\ZZ$-module.

We set the map :
\[\rho: \Lambda \rightarrow \psi(\Lambda/\pf\Lambda)=\FF_q[x;\sigma]/(x^n-u),\] compositum of the canonical projection $\Lambda\rightarrow \Lambda/\pf\Lambda$ with $\psi$. Then    
\[
L= \rho^{-1}(\Cc)=\Ic
\]
is a lattice in $\RR^N$, that is a $\ZZ$-module of rank $N=n^2[F:\QQ]$ since $\Oc_K$ is a $\ZZ$-module of rank $n[F:\QQ]$.

From this point of view, the above construction may be interpreted as a variation of Construction A \cite{splag}, which consists of obtaining a lattice from a linear code over a finite field (ring), as shortly described in the introduction.
This is also a generalization of the lattice construction of \cite{itw}, defined over number fields.\\


\subsection{Examples}
Let $K=\mathbb Q(i)$ and $F=\mathbb Q$. Then $\Oc_F=\mathbb Z$ and $\Oc_K=\mathbb Z[i]$. 
Let $\mathfrak Q$ be the quaternion division algebra defined by 
\[
\mathfrak Q=\mathbb Q(i)\oplus \mathbb Q(i)e,
\] 
with $e^2=-1$. Since $N_{K/F}(a+ib)=a^2+b^2$, $a,b\in\ZZ$, $-1$ cannot be a norm and $\mathfrak Q$ is indeed a quaternion division algebra. 
We set $\Lambda=\mathbb Z[i]\oplus \mathbb Z[i]e$.

We provide different examples based on primes with different ramification. 
\begin{example}\label{ex:p3}\rm
Set $p=3$, which remains inert in $\mathbb Q(i)$. Hence, $\mathbb Z[i]/3\mathbb Z[i]\simeq \FF_9$ and $\Lambda/3\Lambda \simeq \FF_9 \oplus \FF_9 e$. Take $\Ic=(1+i+e)\Lambda$. Then $\Ic$ contains $3$ since the norm of $1+i+e$ is $N(1+i+e)=(1+i+e)(1-i-e)=3$. Let $\alpha$ denote a primitive root of $\FF_9$ over $\FF_3$, satisfying $\alpha^2+1=0$ and note that $\sigma$ becomes the generator of the Galois group of $\FF_{9}/\FF_3$, $\sigma(\alpha)=\alpha^3$. We have 
\[
\psi((1+i+e) {\rm{ mod }} 3)=1+\alpha+x,
\] 
which is a right divisor of $x^2+1$ in $\FF_9[x;\sigma]$: 
\[
(x-1+\alpha)(x+1+\alpha)=x^2+\sigma(1+\alpha)x+(\alpha-1)x+(\alpha+1)(\alpha-1)=x^2+1.
\] 
Therefore, the left ideal $(x+1+\alpha)\FF_9[x;\sigma]/(x^2+1)$ consisting of the left multiples of $x+ 1+\alpha$ modulo $x^2+1$ is a $\sigma$-constacyclic code of length $n=2$ and dimension $k=n-\deg g(x)=1$. 
Left multiples of $x+ 1+\alpha$ modulo $x^2+1$ are explicitly given by $(b_0+b_1x)(x+\alpha+1)=[b_0(\alpha+1)-b_1]+[b_0+b_1(1-\alpha)]x$ and codewords are of the form
\[
(a_0,a_1)=((\alpha+1) a_1, a_1),~a_1 \in \FF_9.
\]
Taking the pre-image by $\psi$, it corresponds to the left-ideal $\Ic/3\Lambda$, with $\Ic=(1+i+e)\Lambda $.
\end{example}

\begin{example}\label{ex:p5}\rm
Set $p=5$, which splits in $\QQ(i)$, namely $(5)=(1+2i)(1-2i)$. Then $\ZZ[i]/5\ZZ[i]\simeq \ZZ[i]/(1+2i) \times \ZZ[i]/(1-2i) \simeq \FF_5 \times \FF_5$ and $\Lambda\simeq (\FF_5\times\FF_5)\oplus (\FF_5\times\FF_5)e$. The first isomorphism comes from the Chinese Remainder Theorem, since $(1-2i)$ and $(1+2i)$ are coprime. An element of $\ZZ[i]$ is mapped to $\FF_5 \times \FF_5$ via the composition of the natural projection and the above isomorphism by $\pi: a+ib \mapsto (a+2b \mod 5,a+3b \mod 5)$. Since $\pi((a+ib)(c+di))=\pi(ac-bd+i(ad+bc))=(ac-bd+2(ad+bc)\mod 5,ac-bd+3(ad+bc)\mod 5)$ while $\pi(a+ib)\pi(c+di)=(a+2b \mod 5,a+3b\mod 5)(c+2d \mod 5,c+3d \mod 5)=(ac+2ad+2cb-bd \mod 5,ac+3ad+3bc-bd \mod 5)$, and $\pi(a+ib+c+id)=a+c+2(b+d)=a+2b+c+2d=\pi(a+ib)+\pi(c+id)$, $\pi$ is indeed a ring homomorphism. Note that $\sigma$ becomes $\sigma(a,b)=(b,a)$, which fixes elements of the form $(a,a)$, $a \in \FF_5$.
Take $\Ic=(1+2e)\Lambda$, which contains $5$ since the norm of $\Ic$ is $N(1+2e)=(1+2e)(1-2e)=5$. Then 
\[
\psi((1+2e) {\rm{ mod }} 5)= (1,1)+(2,2)x,
\] 
which is a divisor of $x^2+1$, since
\[
((1,1)+(2,2)x)((1,1)-(2,2)x)=(1,1)-(2,2)x+(2,2)x+(1,1)x^2.
\]
Then the left ideal $((1,1)+(2,2)x)(\FF_5\times\FF_5)[x;\sigma]/(x^2+1)$ forms a $\sigma$-constacyclic code of length $n=2$ and $k=1$. Codewords are of the form $[(b_0,b_1)+(c_0,c_1)x][(1,1)+(2,2)x]$ or equivalently, $[(b_0,b_1)+(c_0,c_1)x][(3,3)+(1,1)x]=[(3b_0,3b_1)-(c_0,c_1)]+[(b_0,b_1)+(3c_0,3c_1)]x$, that is
\[
(3a,a),~a=(a_0,a_1) \in \FF_5 \times \FF_5.
\]
Again, taking the pre-image by $\psi$, it corresponds to the left-ideal $\Ic/5\Lambda$, with $\Ic=(1+2e)\Lambda $.
\end{example}

\begin{example}\label{ex:p2}\rm
Set $p=2$, which is ramified in $\QQ(i)$: $(2)=(1+i)^2$. Then $\ZZ[i]/2\ZZ[i]\simeq \FF_2 + \upsilon \FF_2 =\{0,1,\upsilon,\upsilon+1\}$ with $\upsilon^2=0$, and $\Lambda/2\Lambda \simeq (\FF_2 + \upsilon \FF_2)\oplus (\FF_2 + \upsilon \FF_2)e $. In this case, the natural projection $\pi$ from $\ZZ[i]$ to $\FF_2 + \upsilon \FF_2$ is given by $\pi:a+ib \mapsto (a+b)\mod 2 + \upsilon (b \mod 2) $. We check that $\pi$ is a ring homomorphism: $\pi((a+ib)(c+di))=\pi(ac-bd+i(ad+bc))=(ac+bd+ad+bc \mod 2)+\upsilon (ad+bc \mod 2)$ while $\pi(a+ib)\pi(c+di)=((a+b \mod 2)+\upsilon(b\mod 2))((c+d\mod 2)+\upsilon(d \mod 2))=(a+b)(c+d)\mod 2+(a+b)d\upsilon+b(c+d)\upsilon=(ac+bd+ad+bc)+\upsilon(ad+bc) \mod 2$, as needed, and $\pi(a+ib+c+id)=a+c+b+d+\upsilon(b+d)=a+b+\upsilon b+c+d+\upsilon d=\pi(a+ib)+\pi(c+id)$. Moreover, since $p=2$ totally ramifies in $\ZZ[i]$, $\sigma$ becomes the identity map.
Take $\Ic=(1+e)\Lambda$, which contains $2$ since the norm of $\Ic$ is $N(1+e)=(1+e)(1-e)=2$. Then 
\[
\psi((1+e) {\rm{ mod }} 2)= 1+x,
\] 
which is a divisor of $x^2+1=(x+1)(x+1)$.
As in the previous cases, the left ideal $(1+x)(\FF_2+v\FF_2)[x;\sigma]/(x^2+1)$ forms a $\sigma$-constacyclic code of length $n=2$ and $k=1$.
Codewords are of the form $(b_0+b_1x)(1+x)$, that is we get the repetition code
\[
(a_0,a_1)=(a_0,a_0),~a_0 \in \FF_2+\upsilon\FF_2.
\] 
As above, taking the pre-image by $\psi$, it corresponds to the left-ideal $\Ic/2\Lambda$, with $\Ic=(1+e)\Lambda $.
\end{example}

%
%

\section{Dual Codes and Lattices}
\label{sec:dualcode}

Having defined constacyclic codes over $(\Oc_K/\pf\Oc_K)[x;\sigma]/(x^n-u)$, it is natural to wonder about the dual of such codes.\\

\subsection{Dual Codes}
We consider the following Euclidean scalar product in $(\Oc_K/\pf\Oc_K)^n$:
\[
\langle y,z \rangle= \sum_{i=1}^n y_i z_i.
\]
\begin{definition}
For an $(n,k)$ code $\Cc$, its Euclidean dual code $\Cc^\perp$ is given by
\[
\Cc^\perp=\{y,~\langle c,y \rangle=0 \mbox{ for all c } \in \Cc\}.
\]
An $(n,k)$ code $\Cc$ is said to be Euclidean self-dual if $\Cc$ is equal to its Euclidean dual.
\end{definition}

It is probably helpful to recall how the generator polynomial of a cyclic code over a finite field is computed. Take $h(x)$ its parity check polynomial, with degree $d$ and constant term $h_0$, and compute $h_0^{-1} x^d h (x^{-1})$. The procedure in the case of interest here is somewhat similar, except for one difficulty: we need to give a meaning to $h(x^{-1})$ in the non-commutative case, which is classically done through localization (as also done in \cite{BUS}).
It is known \cite[Thm II.2.4]{Artin} that given a ring $R$, if $S \subset R$ is a right Ore set, then there is a ring $RS^{-1}$ of right fractions and an injective ring homomorphism $R \rightarrow RS^{-1}$ such that (a) the image of every element $r\in R$ is invertible in $RS^{-1}$, and (b) every element of $RS^{-1}$ can be written as a product $rs^{-1}$.

\begin{proposition}
Consider the ring $(\Oc_K/\pf\Oc_K)[x;\sigma]$ as above, and take $S$ to be $S=\{x^i,~i \geq 0\}$. Then $S$ is a right Ore set, and the right localization $(\Oc_K/\pf\Oc_K)[x;\sigma]S^{-1}$ exists. Furthermore, the subring $A=\{ \sum_{i=0}^d x^{-i}a_i \}$ of $(\Oc_K/\pf\Oc_K)[x;\sigma]S^{-1}$, with multiplication $ax^{-1}=x^{-1}\sigma(a)$ is isomorphic to the ring $(\Oc_K/\pf\Oc_K)[x^{-1},\sigma^{-1}]$.
\end{proposition}
\begin{proof}
By definition of Ore set \cite[Def. II.2.3]{Artin}, we need to check the following three conditions:
\begin{enumerate}
\item
$S$ is closed under multiplication, and $1\in S$.
\item
For any $s \in S$, $s$ is not a left or right zero divisor.
\item
Right Ore condition: for all $r \in R$ and $s\in S$, there exist $r_1 \in R$ and $s_1 \in S$ such that $s r_1=r s_1$. Take $s=x^i \in S$ for some $i$, and $r=r(x)=\sum_{l=0}^d r_l x^l \in \Oc_K/\pf\Oc_K[x;\sigma]$. 
We need to show that there exists $r_1=r_1(x)$ such that $x^i r_1(x)= r(x) x^j$ for some $x^j$. Pick $j \geq i$ and $r_1(x)=\sum_{l=0}^d \sigma^{-i}(r_l)x^{j+l-i}$.
Then
\[
x^i \sum_{l=0}^d \sigma^{-i}(r_l)x^{j+l-i} = \sum_{l=0}^d r_lx^{j+l} = r_1(x)x^j
\]
as needed to show the existence of the localization $(\Oc_K/\pf\Oc_K)[x;\sigma]S^{-1}$.
\end{enumerate}

Next, consider the map
\[
\theta: (\Oc_K/\pf\Oc_K)[x;\sigma] \rightarrow A,~\sum_{i=0}^da_ix^i \mapsto \sum_{i=0}^d x^{-i}a_i.
\]
We have, assuming without loss of generality that $t\geq d$ and $a_{d+1}=\ldots = a_{t}=0$, that 
\[
\theta(\sum_{i=0}^da_ix^i+\sum_{i=0}^t b_ix^i)=
\theta(\sum_{i=0}^t (a_i+b_i)x^i)
=\theta(\sum_{i=0}^da_ix^i)+\theta(\sum_{i=0}^t b_ix^i) 
\]
and
\begin{eqnarray*}
\theta( \sum_{i=0}^da_ix^i\sum_{i=0}^t b_ix^i) & = &
\theta\left(\sum_{k=0}^{d+t} (\sum_{i+j=k} a_i\sigma^i(b_j))x^k \right) \\
&=&
\sum_{k=0}^{d+t} x^{-k} \left(\sum_{i+j=k} a_i\sigma^i(b_j)) \right) \\ 
&=&
\sum_{k=0}^{d+t} \sum_{i+j=k} x^{-j}x^{-i} \sigma^i(b_j) a_i \\
&=&
\sum_{k=0}^{d+t} \sum_{i+j=k} x^{-j} b_j x^{-i} a_i \\
&=&
\theta(\sum_{i=0}^t b_ix^i )\theta(\sum_{i=0}^da_ix^i ) .
\end{eqnarray*}
It follows that $\theta$ is an anti-isomorphism of rings.
\end{proof}

Based on the above proposition, it is now possible to extend the computation of the generator polynomial from the commutative to the non-commutative case.
Recall that we may assume the generator polynomial $g(x)$ to be monic, without loss of generality.

\begin{proposition}\label{prop:gperp}
Consider a code with generator $g(x)$ and parity check polynomial $h(x)=\sum_{i=0}^{k-1}h_ix^i+x^k$, that is $g(x),h(x)$ satisfy $h(x)g(x)=x^n-u$. 
Suppose that $u^2=1$.
The Euclidean dual $g^\perp(x)/(x^n-u)$ of the $\sigma$-constacyclic code $(g(x))/(x^n-u)$ is a $\sigma$-constacyclic code whose generator polynomial is given by
\[
g^\perp(x)=1+\sum_{i=1}^k \sigma^{i}(h_{k-i})x^i,
\] 
or equivalently by the monic polynomial $-ug^\perp(x)$.
\end{proposition}
\begin{proof}
To show that $g^\perp(x)/(x^n-u)$ is a $\sigma$-constacyclic code, we need to show that $g^\perp(x)$ is a right divisor of $x^n-u$. We know that $g(x)$ is, namely there exists a polynomial $h(x)$ such that
\[
g(x)h(x)=h(x)g(x)=x^n-u.
\]
Then, using the anti-isomorphism $\theta: (\Oc_K/\pf\Oc_K)[x;\sigma] \rightarrow A,~\sum_{i=0}^da_ix^i \mapsto \sum_{i=0}^d x^{-i}a_i$ of the previous proposition:
\begin{eqnarray*}
x^{k}\theta(h(x))\theta(g(x))x^{n-k}&=& x^{k}\theta(g(x)h(x))x^{n-k}\\
&=&x^{k}\theta(x^n-u)x^{n-k} \\
&=& x^{k}(x^{-n}-u)x^{n-k} \\
&=& 1-ux^n = -u (x^n-\frac{1}{u}) =-u (x^n-u),
\end{eqnarray*}
since $u^2=1$. Therefore $x^k\theta(h(x))$ is a left divisor of $x^n-u$, and
\[
x^k\theta(h(x)) = x^k\theta(\sum_{i=0}^{k-1} h_ix^i+x^k)
=x^k(\sum_{i=0}^{k-1} x^{-i} h_i+x^{-k})  = \sum_{i=0}^{k-1}\sigma^i(h_i) x^{k-i}+1.
\]
By the first part of Proposition \ref{prop:parity}, $x^k\theta(h(x))$ is also a right divisor of $x^n-u$.
Now that we have shown that $g^\perp(x)$ is indeed a $\sigma$-constacyclic code, 
we are left to show that it is the dual of $\Cc$.\\

Let $a(x)=\sum_{i=0}^{n-1}a_ix^i$ be the polynomial corresponding to the codeword $a=(a_0,a_1,\ldots,a_{n-1})$. Then we denote, for any $0\le l\le n-k-1$, by $d_l(x)$ the polynomial $x^lg^{\perp}(x)$, which corresponds to the codeword :
\[
d_l=( 0_{1\times l},\sigma^l(h_k),\sigma^{l+1}(h_{k-1}),\ldots,\sigma^{l+k}(h_0),0_{1\times(n-k-l+1)}),~
l=0,\ldots,n-k-1.
\]

Hence, $a$ is orthogonal to the $n-k$ vectors $d_l$ ($0\le l\le n-k-1$): indeed, note that
\[
d_l a^T=\sum_{i=l}^{l+k}a_i\sigma^i(h_{l+k-i})
\]
which turns out to be the coefficient of $x^{k+l}$ in $a(x)h(x)$, $l=0,\ldots,n-k-1$. Since $\deg h(x)=k$, $\deg(a(x)h(x))<n+k$, say $a(x)h(x)=p(x)=\sum_{i=0}^{n+k-1}p_ix^i$, which means that modulo $x^n-u$, $$a(x)h(x)=\sum_{i=0}^{n-1}p_ix^i+\sum_{i=n}^{n+k-1}p_ix^i=\sum_{i=0}^{n-1}p_ix^i+\sum_{j=0}^{k-1}p_{j+n}x^{j}(-u).$$ 

Moreover, since $a\in \mathcal C$, by Proposition 1, we have $a(x)h(x)=0$ modulo $x^n-u$, implying that the coefficients of $x^{k+l}$, $l=0,\ldots,n-k-1$, are equal to zero. This proves the orthogonality of $d_l$ with respect to $a$.

Therefore, the $\Oc_K/\mathfrak p\Oc_K$-module generated by the codewords $d_l$ ($0\le l\le n-k-1$) is contained in $\mathcal C^\perp$.

We now show that the module $\Oc_K/\mathfrak p\Oc_K[d_0,...,d_{n-k-1}]$ has rank $n-k$. Indeed, let $\lambda_0,...,\lambda_{n-k-1}\in \Oc_K/\mathfrak p\Oc_K$ such that $\underset{l=0}{\overset{n-k-1}{\sum}} \lambda_ld_l=0$.  The first coordinate of $\underset{l=0}{\overset{n-k-1}{\sum}} \lambda_ld_l$ is $\lambda_0h_k=\lambda_0$, so $\lambda_0=0$. Then $\underset{l=1}{\overset{n-k-1}{\sum}} \lambda_l d_l=0$ and by the same argument, we prove successively that $\lambda_1=0$, ..., $\lambda_{n-k-1}=0$. This proves that the $n-k$ vectors $d_0,...,d_{n-k-1}$ are linearly independent over $\Oc_K/\mathfrak p\Oc_K$.
\end{proof}

We remark that in our context of cyclic division $F$-algebras, $u$ is restricted to $F$, therefore the polynomial $x^n - u \in (\Oc_K/p\Oc_K)[x;\sigma]$ actually has a constant term $u$ which belongs to $\Oc_F /\pf\Oc_F$. Therefore the condition $u^2 = 1$ means that $u$ is of order $2$ in a finite field, that is $u = \pm 1$. This is however not the case more generally.

\begin{example}\label{ex:p3_2}\rm
Let us continue Example \ref{ex:p3}.
We have $x^2+1=(x+\alpha+1)(x+\alpha-1)$ therefore the parity check polynomial is 
$h(x)=x+\alpha-1$. From Proposition \ref{prop:gperp}, we have that $g^\perp(x)=1+\sigma(h_0)x=1+\sigma(\alpha-1)x=1-(\alpha+1)x$. Now left multiple of $g^\perp(x)$ modulo $x^2+1$ are of the form $(b_0+b_1x)(x-(\alpha+1))=[b_0-b_1(\alpha-1)]+[b_0(-\alpha-1)+b_1]x$, so that codewords are 
\[
(a_0,a_1)=(a_0,(-\alpha-1)a_0).
\] 
These codewords clearly form the dual code of $\Cc$, since 
$\langle (a_0,(-\alpha-1)a_0) , y \rangle =0$ for all $y \in \Cc$.
\end{example}

\begin{example}\rm
For Example \ref{ex:p5}, we have $((1,1)+(2,2)x)((1,1)-(2,2)x)=(1,1)+(1,1)x^2$ or equivalently $((3,3)+(1,1)x)((2,2)-(4,4)x)=(1,1)+(1,1)x^2$ with parity check polynomial $h(x)=(2,2)+(1,1)x$. From Proposition \ref{prop:gperp}, we have that $g^\perp(x)=1+\sigma(h_0)x=1+\sigma((2,2))x=1+(2,2)x$. Now left multiple of $g^\perp(x)$ modulo $x^2+1$ are of the form $(b_0+b_1x)(1+(2,2)x)=[b_0-b_1(2,2)]+[b_0(2,2)+b_1]x$, so that codewords are 
\[
(a_0,a_1)=(a_0,(2,2)a_0),~a_0 \in \FF_5 \times \FF_5,
\]
which form the dual code of $\Cc$, since 
$\langle (a_0,(2,2)a_0) , y \rangle =0$ for all $y \in \Cc$.
This also shows that $\Cc$ is self-dual.
\end{example}

Finally we provide another example of self-dual code.

\begin{example}\label{ex:sd}\rm
Let $K=\mathbb Q(\sqrt{2})$ and $F=\mathbb Q$. Then $\Oc_F=\mathbb Z$ and $\Oc_K=\mathbb Z[\sqrt{2}]$. 
Let $\mathfrak Q$ be the quaternion division algebra defined by 
\[
\mathfrak Q=\mathbb Q(\sqrt{2})\oplus \mathbb Q(\sqrt{2})e,
\] 
with $e^2=-5$. Since $N_{K/F}(a+\sqrt{2}b)=a^2+2b^2$, $a,b\in\ZZ$, $-5$ cannot be a norm and $\mathfrak Q$ is indeed a quaternion division algebra. 
We set $\Lambda=\mathbb Z[\sqrt{2}]\oplus \mathbb Z[\sqrt{2}]e$.

Set $p=3$, which remains inert in $\mathbb Q(\sqrt{2})$. Hence, $\mathbb Z[\sqrt{2}]/3\mathbb Z[i]\simeq \FF_9$ and $\Lambda/3\Lambda \simeq \FF_9 \oplus \FF_9 e$. Take $\Ic=(\sqrt{2}+e)\Lambda$. Then $\Ic$ contains $3$ since the norm of $\sqrt{2}+e$ is $N(\sqrt{2}+e)=(\sqrt{2}+e)(-\sqrt{2}-e)=3$. Let $\alpha$ denote a primitive root of $\FF_9$ over $\FF_3$, satisfying $\alpha^2+1=0$ and as before, $\sigma(\alpha)=\alpha^3$. We have 
\[
\psi((\sqrt{2}+e) {\rm{ mod }} 3)=\alpha+x,
\] 
which is a right divisor of $x^2+5= x^2+2$ in $\FF_9[x;\sigma]$: 
\[
(x+\alpha)(x+\alpha)=x^2+\sigma(\alpha)x+(\alpha)x+\alpha^2=x^2+2.
\] 
Therefore, the left ideal $(x+\alpha)\FF_9[x;\sigma]/(x^2+2)$ is a $\sigma$-constacyclic code of length $n=2$ and dimension $1$. 
Left multiples of $x+\alpha$ modulo $x^2+2$ are explicitly given by $(b_0+b_1x)(x+\alpha)=[b_0(\alpha)+b_1]+[b_0+b_1(-\alpha)]x$ and codewords are of the form
\[
(a_0,a_1)=(\alpha a_1, a_1),~a_1 \in \FF_9.
\]
This code is self-dual, since $u=2$ satisfies $u^2=1$ and $g^\perp(x)=1+\sigma(\alpha)x=1-\alpha x=-\alpha(\alpha+x)$.
Taking the pre-image by $\psi$, it corresponds to the left-ideal $\Ic/3\Lambda$, with $\Ic=(\sqrt{2}+e)\Lambda $.

\end{example}

\subsection{Lattices from Dual Codes}

Recall the proposed variation of Construction A. Given the map
\[
\rho: \Lambda \rightarrow \psi(\Lambda/\pf\Lambda)=\FF_q[x;\sigma]/(x^n-u),
\] 
compositum of the canonical projection $\Lambda\rightarrow \Lambda/\pf\Lambda$ with $\psi$, we obtained a lattice $L$ given by    
\[
L= \rho^{-1}(\Cc)=\Ic.
\]
Now let $\Cc^{\perp}$ be the dual of $\Cc$. Then
\[
L'=\rho^{-1}(\Cc^\perp)
\]
is also a lattice.

\begin{example}\rm
Let us continue Examples \ref{ex:p3} and \ref{ex:p3_2}.
To the code $\Cc$ given by the left multiples of $(x+1+\alpha) \mod x^2+1$ corresponds the lattice $\rho^{-1}(\Cc)=(1+i+e)\Lambda$. 
The dual lattice $\Cc^\perp$ is generated by the left multiples of $1-(\alpha+1)x \mod x^2+1$ to which corresponds the lattice $\rho^{-1}(\Cc^\perp)=(1-(i+1)e)\Lambda$. 
\end{example}

We make an obvious observation.

\begin{lemma} 
If $\mathcal C\subset \mathcal C^\perp$, we have the inclusion $L\subset L'$. 
\end{lemma}
\begin{proof} 
Let $x\in L$. Then $\rho(x) \in \mathcal C$, so $\rho(x)\in C^\perp$ and $x\in L'$. 
\end{proof}

%
%
%

\section{Application to Space-time Codes}
\label{sec:matrix}

We conclude by discussing our motivation to look at the question of defining a variation of Construction A over a cyclic division $F$-algebra, which comes from space-time coding.

\subsection{Space-Time Coding}
Cyclic division algebras are by now classically used to design space-time codes \cite{SSV,BO}. 
Unlike in classical coding theory where codewords are typically vectors with coefficients over finite fields (or finite rings), in the context of space-time coding, codewords are matrices with coefficients coming from number fields, and matrix codewords are obtained as follows. To make the notation easier, we assume for the rest of this section that $u\in \Oc_F$.
To any element $a= a_0+a_1e+\cdots +a_{n-1}e^{n-1}$ of $\Lambda$, we can associate a matrix in ${\rm{Mat}}_n(\Oc_K)$ (since $u\in \Oc_F$) by :
\[ 
M(a)= 
\begin{bmatrix}
a_0 & u\sigma(a_{n-1}) & u\sigma^2(a_{n-2}) & \cdots & u \sigma^{n-1}(a_1) \\ 
a_1 &    \sigma(a_{0}) & u\sigma^2(a_{n-1}) & \cdots & u \sigma^{n-1}(a_2) \\ 
\vdots & \vdots & \vdots & \ddots & \vdots \\ 
\vdots & \vdots & \vdots & & u\sigma^{n-1}(a_{n-1}) \\ 
a_{n-1} & \sigma(a_{n-2}) & \sigma^2(a_{n-3}) & \cdots &  \sigma^{n-1}(a_0) \end{bmatrix}.\]

The map \[\begin{aligned} \Lambda &\rightarrow {\rm{Mat}}_n(\Oc_K)\\ a & \mapsto M(a)\end{aligned}\] is an $\Oc_K$-algebra injective homomorphism. 
A space-time code is thus given by $\{ M(a),~a\in \Lambda\}$. The condition that the $F$-algebra is division is critical, it fulfills one design criterion for space-time codes, namely that $\det(M(a)-M(a'))\neq 0$, $a \neq a'$.

We illustrate this by continuing Example \ref{ex:p3}.
\begin{example}\rm
For $q=a+be$ in the natural order $\ZZ[i]\oplus\ZZ[i]e$ of the quaternion algebra $\mathfrak{Q}$, $a,b \in \ZZ[i]$
\[
M(q)=
\begin{bmatrix}
a  & -\bar{b} \\
b  & \bar{a}\\
\end{bmatrix}
\]
where $\bar{\cdot}$ is the non-trivial Galois automorphism of $\QQ(i)/\QQ$.
Let $t=(a+be)(1+i+e)$ be an element of $\Ic=\Lambda(1+i+e)$ (with $a,b\in \mathbb Z[i]$). Then 
\[t=a(1+i)-b + (a+b(1-i))e.\]
Hence, 
\[
M(t)=\begin{bmatrix} a(1+i)-b & -(\overline{a} +\overline{b}(1+i))\\ a+b(1-i) & \overline{a}(1-i)-\overline{b} \end{bmatrix},
\]
and we obtain the space-time code 
\[
\left\{
\begin{bmatrix} a(1+i)-b & -(\overline{a} +\overline{b}(1+i))\\ a+b(1-i) & \overline{a}(1-i)-\overline{b} \end{bmatrix},~a,b \in \ZZ[i]
\right\}.
\]

Note that $\Ic=\rho^{-1}(C)$ is a real lattice with rank $4$. 
\end{example}

\subsection{Coset Coding.}
Let now $v=(v_1,\ldots,v_n)$ be an information vector containing the data to be transmitted, which should be mapped to a lattice point in $L$, where $L$ is used as a lattice code. Mapping a vector with coefficients in a finite ring to a lattice point is not an easy task, and can be facilitated by the use of Construction A.
The lattice $L=\rho^{-1}(\Cc)=\Ic \Lambda$ may by construction be written 
as a union of cosets of $\pf\Lambda$, where each coset representative may be chosen to be a codeword in the code $\Cc$. 
Namely, if $g(x)$ is a right divisor of $x^n-u$ and if a $\sigma$-constacyclic code $\Cc=(g(x))/(x^n-u)\subset (\Oc_K/\pf\Oc_K)[x;\sigma]/(x^n-u)$ has dimension $k=n-\deg(g)$, since 
\[
\Lambda/\pf\Lambda \cong (\Oc_K/\pf\Oc_K)[x;\sigma]/(x^n-u)
\] there is an isomorphism 
\[
\Ic/\pf\Lambda \cong \Cc.
\] 
This allows us to associate in a unique way a coset of $\pf\Lambda$ to a codeword.
The mapping from $v$ to a point in $L$ may be done by attributing some information coefficients $v_1,\ldots,v_k$ to be encoded using the code $\Cc$, and the rest of the information coefficients to be mapped to a point in the lattice $\pf\Lambda$.
This simplifies the encoding in the cases where $\pf\Lambda$ is a lattice ``easier" to label than $L$, which happens in the original Construction A where this lattice turns out to be a scaled version of $\ZZ^n$, but also in many cases where $\pf$ gives rise to a lattice isomorphic to $\ZZ^n$ (or another lattice whose points are easy to label).

Irrespectively of the lattice obtained from $\pf\Lambda$, coset encoding is necessary in the context of wiretap codes. In a wiretap context, a code should not only provide reliability between a transmitter and a receiver, but also ensure confidentiality, should an eavesdropper try to intercept the message. To protect from wiretapping, information symbols are mapped to a codeword in $\Cc$, while random symbols are picked uniformly at random in the lattice $\pf\Lambda$ to create confusion at the eavesdropper. Wiretap space-time codes have been studied in \cite{OB}, where coset encoding of space-time codes is assumed. However, no concrete way to do so was proposed.
The construction of the lattice $L=\rho^{-1}(\Cc)=\Ic$ as presented in this work thus enables coset encoding for wiretap space-time codes.

%
%
%
\section{Future Work}

In this paper, we presented a construction of lattices from constacyclic codes over skew-polynomials, which can be interpreted as a variation of the well known Construction A of lattices from linear codes. The starting point was the quotient $\Lambda/\pf\Lambda$ of a given order $\Lambda$ from a cyclic division $F$-algebra, which is motivated by applications of such algebras to space-time coding.
Natural future research directions include:
\begin{itemize}
\item
To continue the study of constacyclic codes over $(\Oc_K/\pf\Oc_K)[x;\sigma]/(x^n-u)$, in particular by replacing the polynomial $x^n-u$ with a more general polynomial $f(x)$. This may include looking at the characterization of self-dual codes. One could alternatively consider duality with respect to a Hermitian inner product.
\item
Linking the properties of the constacyclic code $\Cc$ to that of the lattice $L=\rho^{-1}(\Cc)$: there are standard duality results for the classical Construction A, relating the dual lattice of $L^*$ with $L'$, or the weight enumerator of the code with the theta series of the lattice, which have not yet been considered. This leads to further properties of lattices that could be explored, such as extremality and modularity.
\item
Design of wiretap space-time codes: this consists of choosing the cyclic division $F$-algebra $A$, the corresponding two-sided ideal $\Ic$ and constacyclic code $\Cc$, to optimize the confusion at the eavesdropper.
\end{itemize}

%
%


\begin{thebibliography}{15}
%
\bibitem{splag}
J.H. Conway and N.J.A Sloane, ``Sphere Packings, Lattices and Groups", Springer.
%
\bibitem{Forney} G. D. Forney, “Coset Codes − Part I: Introduction and geometrical classification”, {\em IEEE Trans. on Inform. Theory}, vol 34, no 5, 1988.
%
\bibitem{Ebeling} W. Ebeling, “Lattices and Codes, A Course Partially Based on Lectures by Friedrich Hirzebruch”, Springer, in the series Advanced Lectures in Mathematics.
%
\bibitem{Bachoc} C. Bachoc, “Applications of Coding Theory to the Construction of Modular Lattices”, {\em Journal of Combinatorial Theory}, 1997.
%
\bibitem{BSC} A. Bonnecaze, P. Sol´e, A.R. Calderbank, “Quaternary Qudratic Residue codes and Unimodular Lattices”, {\em IEEE Trans. on Information Theory}, vol. 41, no 2, March 1995.
%
\bibitem{Boucher}
D. Boucher and F. Ulmer, ``Coding with skew polynomial rings", {\em Journal of Symbolic Computation}, vol. 44, 2009.
%
\bibitem{BGU}
D. Boucher, W. Geiselmann and F. Ulmer, Skew Cyclic Codes, {\em Applied Algebra in
Engineering, Communication and Computing}, 18 (2007), 379 - 389
%
\bibitem{BUS}
D. Boucher, F. Ulmer, P. Sol\'e, ``Skew Constacyclic Codes over Galois Rings"
{\em Advances in Mathematics of Communications}, 2,  273-292 (2008).
%
\bibitem{OS}
F. Oggier and B. A. Sethuraman, ``Quotients of Orders in Cyclic Algebras and Space-Time Codes", 
{\em Advances in Mathematics of Communication}, vol. 7, November 2013, 441-461.
%
\bibitem{castle}
J. Ducoat, F. Oggier, ``Lattice Encoding of Cyclic Codes from Skew-polynomial Rings", 
proceedings of the Fourth International Castle Meeting on Coding Theory and Applications, Palmela, 2014. 
%
\bibitem{itw}
F. Oggier and J.-C. Belfiore, ``Enabling Multiplication in Lattice Codes via Construction A", 
proceedings of the IEEE International Workshop on Information Theory, Sevilla, 2013.
%
\bibitem{Artin}
M. Artin,``Noncommutative Rings", 1999.
%
\bibitem{SSV}
B. A. Sethuraman and B. S. Rajan and V. Shashidhar, ``Full-diversity, high-rate space-time block codes from division algebras", {\em IEEE Trans. on Inform. Theory}, vol. 49, no 10, 2003.
%
\bibitem{BO}
G. Berhuy and F. Oggier, ``An Introduction to Central Simple Algebras and Their Applications to Wireless Communication", AMS, 2013.
%
\bibitem{OB}
J.-C. Belfiore and F. Oggier, ``An Error Probability Approach to {MIMO} Wiretap Channels",
{\em IEEE Transactions on Communications}, vol. 61, no. 8, 2013.
%
\end{thebibliography}
\end{document}